\documentclass[11pt]{article}

\usepackage{typearea}
\typearea{13}

\usepackage{latexsym,graphicx}
\usepackage{amssymb,amsthm,mathtools,bbm}
\usepackage{xspace}
\usepackage{bm}
\usepackage{ifpdf}
\usepackage{nicefrac}
\usepackage{shadow,shadethm,color}
\usepackage[linesnumbered,ruled,lined,noend]{algorithm2e}
\usepackage{subfigure}
\usepackage{verbatim}
\usepackage{xcolor}
\usepackage[normalem]{ulem}
\usepackage{thm-restate}

\usepackage{natbib}
\usepackage{enumitem}
\usepackage{fullpage}
\usepackage{parskip}

\allowdisplaybreaks

\definecolor{Darkblue}{rgb}{0,0,0.4}
\definecolor{Brown}{cmyk}{0,0.81,1.,0.60}
\definecolor{Purple}{cmyk}{0.45,0.86,0,0}
\newcommand{\mydriver}{hypertex}
\ifpdf
 \renewcommand{\mydriver}{pdftex}
\fi
\usepackage[breaklinks,\mydriver]{hyperref}
\hypersetup{colorlinks=true,
            citebordercolor={.6 .6 .6},linkbordercolor={.6 .6 .6},
citecolor=blue,urlcolor=black,linkcolor=blue,pagecolor=black}
\newcommand{\lref}[2][]{\hyperref[#2]{#1~\ref*{#2}}}

\usepackage[nameinlink]{cleveref}
\Crefname{algocf}{Algorithm}{Algorithms}
\crefname{algocfline}{line}{lines}
\Crefname{invariant}{Invariant}{Invariants}
\Crefname{claim}{Claim}{Claims}
\Crefname{subclaim}{Subclaim}{Subclaims}

\newtheorem*{theorem*}{Theorem}
\newtheorem{theorem}{Theorem}[section]
\newtheorem{definition}[theorem]{Definition}

\newtheorem{lemma}[theorem]{Lemma}
\newshadetheorem{lemmashaded}[theorem]{Lemma}
\newtheorem{fact}[theorem]{Fact}

\newtheorem{claim}[theorem]{Claim}

\newtheorem{corollary}[theorem]{Corollary}

\newcommand{\junk}[1]{}
\newcommand{\ignore}[1]{}

\newcommand{\R}[0]{{\ensuremath{\mathbb{R}}}}

\newcommand{\Z}[0]{{\ensuremath{\mathbb{Z}}}}

\newcommand{\argmin}{\operatorname{argmin}}

\newcommand{\sse}{\subseteq}

\newcommand{\cA}{{\mathcal{A}}}

\newcommand{\cH}{{\mathcal{H}}}

\newcommand{\cR}{{\mathcal{R}}}
\newcommand{\cS}{{\mathcal{S}}}

\newcommand{\cU}{{\mathcal{U}}}
\newcommand{\cV}{{\mathcal{V}}}
\newcommand{\cX}{{\mathcal{X}}}

\newcommand{\EE}{{\mathbb{E}}\,}

\newcommand{\ones}{\bm{1}}

\newcommand{\bigblacktriangleup}{\scalebox{1.3}{$\blacktriangle$}}
\newcommand{\fullsimplex}{\bigblacktriangleup}
\newcommand{\scalesmpx}{\overline \fullsimplex}

\newcommand{\OPT}{\textup{\textsc{Opt}}\xspace}
\newcommand{\ALG}{\textup{\textsc{Alg}}\xspace}
\newcommand{\AUG}{\textup{\textsc{Aug}}\xspace}
\newcommand{\LP}{\textup{\textsc{LP}}\xspace}
\newcommand{\coins}{\omega}
\newcommand{\augoco}{\textup{\textsc{Mono-OCO}}\xspace}
\newcommand{\stococo}{\textup{\textsc{OCO-Alg}}\xspace}
\newcommand{\gain}{\ensuremath{{\sf gain}}}
\newcommand{\realgain}{\ensuremath{{\sf Rgain}}}
\newcommand{\aug}{\ensuremath{{\sf aug}}}

\newcommand{\ip}[1]{\langle #1 \rangle}

\newcommand{\E}[2]{\underset{#1}{\mathbb{E}}\left[#2\right]}

\newcommand{\ind}[1]{\mathbbm{1}\!\left(#1\right)}

\newcommand{\bs}[1]{\boldsymbol{#1}}

\newcounter{note}[section]

\newcommand{\nf}{\nicefrac}

\newcommand{\Suffix}{\cU}
\newcommand{\Uncov}{\cV}

\newcommand{\optest}{\textsc{Est}}
\newcommand{\approxfactor}{\alpha}
\newcommand{\density}{\rho}

\newcommand{\eat}[1]{}

\newcommand{\LoC}{\textsc{LearnOrCover}\xspace}

\usepackage{amsmath}
\usepackage{xcolor}
\usepackage{mathtools}

\begin{document}

\title{A Learning Perspective on Random-Order Covering Problems}
    \author{Anupam Gupta\thanks{New York University, USA (anupam.g@nyu.edu).}
    \and Marco Molinaro\thanks{Microsoft Research -- Redmond, USA and PUC-Rio, Brazil
  (mmolinaro@microsoft.com).}
  \and Matteo Russo\thanks{Sapienza University of Rome, Italy
  (mrusso@diag.uniroma1.it).}}

\date{}

\maketitle

\thispagestyle{empty}

\begin{abstract}
  In the random-order online set cover problem, the instance with $m$
  sets and $n$ elements is chosen in a worst-case fashion, but then
  the elements arrive in a uniformly random order. Can this
  random-order model allow us to circumvent the bound of
  $O(\log m \log n)$-competitiveness for the adversarial arrival order
  model? This long-standing question was recently resolved by
  \cite{GuptaKL21}, who gave an algorithm that achieved an
  $O(\log mn)$-competitive ratio. While their \LoC was inspired by
  ideas in online learning (and specifically the multiplicative
  weights update method), the analysis proceeded by showing progress from
  first principles.

  \medskip
  In this work, we show a concrete connection between random-order set
  cover and stochastic mirror-descent/online convex optimization. In
  particular, we show how additive/multiplicative regret bounds for
  the latter translate into competitiveness for the former.
  Indeed, we give a clean recipe for this translation, allowing us to
  extend our results to covering integer programs, set multicover, and
  non-metric facility location in the random order model, matching
  (and giving simpler proofs of) the previous applications of the \LoC
  framework.
\end{abstract}

\newpage
\pagenumbering{arabic} 

\section{Introduction}
\label{sec:introduction}

In the online \emph{set cover} problem, the adversary chooses a set system
$(\cU, \cS)$, and reveals the elements of the universe one by
one. Upon seeing an element $e \in \cU$, the algorithm learns which
sets contain this element, and must ensure that $e$ is covered, i.e.,
it has picked at least one of the sets containing element $e$. The goal is to
pick the fewest sets, or the cheapest collection of sets, if sets have
non-negative costs. The offline version of this problem admits a
$(1 + \ln n)$-approximation, where $n = |\cU|$, and this is best possible
(up to lower order terms) unless $\text{P}=\text{NP}$ (see, e.g.,
\cite{DBLP:books/daglib/0030297}). However, the algorithms that
achieve this approximation guarantee require knowing the set system
up-front. What can we do in an online setting?

In a landmark paper, Alon et
al.~\cite{DBLP:journals/siamcomp/AlonAABN09} gave an
$O(\log m \log n)$-competitive randomized online algorithm; they also
showed that one cannot do much better using deterministic algorithms,
even if the set structure was known up-front, but only a subset of the
elements of $\cU$ would arrive online. Subsequently, Korman~\cite{korman2004use} showed that no polynomial-time algorithms could
beat this double-logarithmic bound.  The question of doing better
beyond the worst-case was considered soon thereafter: Grandoni et
al.~\cite{DBLP:journals/siamcomp/GrandoniGLMSS13} gave an
$O(\log mn)$-competitive algorithm if the set system was known, and
the requests were drawn i.i.d.\ from an fixed distribution over the
elements of $\cU$.

The question of extending it to the random-order model remained open
for considerably longer. It was finally resolved when an
$O(\log mn)$-competitive algorithm was given even when the set system
was fixed and unknown, but the elements arrived in uniformly random
order~\cite{GuptaKL21}.  They also generalized their techniques to
covering integer programs without box-constraints, set multicover, and
non-metric facility location~\cite{GuptaKL21,GuptaKL24}.  These
random-order algorithms (whom we refer to as the \LoC family) were all
based on the multiplicative weights update method
to learn the optimal solution using the random-order samples. They
were inspired by stochastic gradient/mirror descent---that of using the
elements arriving in random order---as random samples, from which to
learn a good solution.

However, this connection between \LoC algorithms and online learning
was more in spirit than at a technical level: the proofs in
\cite{GuptaKL21,GuptaKL24} proceeded from first-principles, showing
that the KL divergence to the optimal solution plus the (logarithm of
the) number of uncovered elements decreased in expectation fast
enough. Again, the guiding principle was that these algorithms either
make progress towards learning the optimal solution, or towards
covering elements, at each step. However, each problem required a
somewhat involved potential-function calculation, and it was difficult
to see the intuition for how the algorithmic decisions were informing
the technical details and the convergence proofs.

\subsection{Our Results}

In this work, we make the connection between random-order online algorithms and online learning—particularly regret minimization—explicit, modular, and arguably conceptually cleaner. We isolate the online learning component, showing that any (stochastic) online convex optimization algorithm with suitable additive/multiplicative regret bounds can be used as a black box subroutine to yield an online algorithm with optimal competitiveness guarantees.

This general approach yields the following unified result:

\begin{theorem*}
There exist $O(\log mn)$-competitive algorithms for the following random-order covering problems:
    \begin{enumerate}[nosep,label=(\roman*)]
        \item (weighted) Set Cover,
        \item (unweighted) Set Multicover,
        \item Covering Integer Programs (without box constraints), and
        \item Non-metric Facility Location.
    \end{enumerate}
\end{theorem*}

While prior work \citep{GuptaKL21, GuptaKL24} established the existence of $O(\log mn)$-competitive algorithms for these problems, our contribution is a conceptual simplification: \emph{the same} OCO algorithm---used entirely as a black box---applies uniformly across all settings. The only problem-specific component are the concave gain functions fed to the OCO routine.

\subsection{Our Techniques}
\label{sec:our-techniques}

To illustrate the ideas, consider the case of unweighted online set
cover, in the random-order model. Recall that while the adversary
fixes the set system $(\cU, \cS)$ up-front, we know nothing apart from
$m$, the number of sets in $\cS$, at the beginning. (For this
particular discussion, assume also that we know $|\OPT|$, the cost of the
optimal solution.) Crucially, we do not know which elements belong to
which sets: only when an element arrives we known which sets contain it. 

As mentioned above, our algorithms follow the \LoC framework
of~\cite{GuptaKL21}. In this framework, at each timestep $t$ we maintain a probability distribution
$q^t \in [0,1]^m$ over the sets in $\cS$. If the arriving element $e^t$ is not yet covered, we add to our solution any
set that covers it, and, crucially, sample and add to our solution one
more set according to the distribution $q^t$. The intuition for the sampling is that,
by paying a cost of two (sets) instead of one in this iteration, we will help
covering \emph{future} elements.

The key is how to choose---or indeed, ``learn''---the distribution
$q^t$, after we have seen $t-1$ elements from $\cU$, given that we do
not know which sets cover the future elements. Let $\Uncov^t$
denote the elements of $\cU$ which are yet to arrive and have not been
already covered by sets chosen at steps $1, \ldots, t-1$. Suppose we
define
$\gain_{t}(q) = \sum_{S \in \cS} q_S \cdot \frac{\textrm{\# elements $S$ covers
    from $\Uncov^t$}}{|\Uncov^t|}$; this measures exactly the fraction
of uncovered future elements that we expect to cover if we sample a
set using the distribution $q$. Note that, since we do not know the set
structure, we do not know the function $\gain_t(q)$---but more about
that soon.

Now, if we magically knew the optimal solution $\OPT \sse \cS$, we
could use the uniform distribution
$q^\star := \frac{\bs{1}(S \in \OPT)}{|\OPT|}$ over $\OPT$'s sets and
get that $\gain_{t}(q^\star) \ge \frac{1}{|\OPT|}$. And hence, after
about $\log n \cdot |\OPT|$ steps sampling from this distribution, we
would reduce the fraction of uncovered elements to less than
$(1 - \frac{1}{|\OPT|})^{(\log n) |\OPT|} \le \frac{1}{n}$, i.e., all
$n$ elements would be covered!

There are two obvious issues with this thought experiment:
\begin{enumerate}
\item We do not know the optimal solution, and hence the fractional
  solution $q^\star$. The key idea here is to use \emph{online convex
    optimization} to learn $q^\star$. More precisely, OCO allows us to
  choose a sequence of distributions $\{q^t\}_{t=1}^n$ that perform
  almost as well as $q^\star$ for the expected coverage; namely
  \[ \sum_t \gain_{t}(q^t) \gtrsim \frac{1}{2} \sum_t \gain_{t}(q^\star) -
    O(\log m); \] we use a multiplicative/additive regret
  guarantee to avoid the typical dependence of $\sqrt{n}$. 
\item We do not know the function $\gain_t(\cdot)$, since we do not
  know which sets contain which future elements and hence we don't
  know $\Uncov^t$.)  But here's the second key observation: if we see
  a random uncovered element $e^t$ at timestep $t$, we can use it to
  compute an unbiased stochastic approximation of (the gradient of)
  this function $\gain_t$; this is where we rely on the random-order
  model. More precisely, the linear function
  \[ \gain_{t,e^t}(q) = \sum_S q_S \cdot \bs{1}(e^t \in \Uncov^t
    \textrm{ and } e^t \in S)\] satisfies
  $\EE_{e^t} [\,\gain_{t,e^t}(q)\,] = \gain_{t,\Uncov^t}(q)$; moveover, when
  $e^t$ arrives, we know whether it is uncovered, and which sets cover
  it, so we can compute $\gain_{t,e^t}$ using the information at hand,
  and feed this stochastic gradient to the OCO algorithm.
\end{enumerate}
\medskip

While the approach sketched above is true in spirit to our algorithm
(which is essentially the same for all problems),
the details have some crucial technical differences. Firstly, instead
of maintaining a probability distribution $q^t$ over sets/resources,
we maintain a vector $p^t \in [0,1]^m$ (whose coordinates do not need to
add up to 1), and may sample multiple sets/resources while controlling for
their expected costs. Secondly, instead of tracking the
number/fraction of uncovered future elements, we use a potential
function that takes into account the cost of covering these
elements. Finally, the ``real'' gain, which is the actual decrease in the
potential, is somewhat complicated and so we have to ``linearize''
before using it in our OCO algorithm. Moreover, to ensure that these
OCO algorithms suffer low-regret, it is crucial that this linearized
gain functions have bounded gradients.

\medskip\textbf{Paper Organization:} In \Cref{sec:warmup} we give our
algorithm for general set cover; the algorithm and the proof of
correctness captures most of the nuances sketched in the previous
paragraph. We then abstract out the framework in \Cref{sec:framework},
and apply this framework in \Cref{sec:applications} to unweighted set
multicover, to covering integer programs, and to the non-metric
facility-location problem.

\subsection{Related Work}
\label{sec:related-work}

Random order models have seen considerable research over the years;
see \cite{DBLP:books/cu/20/Gupta020} for a survey and the historical
perspective. The problems that have been considered include set
cover~\cite{GuptaKL21,GuptaKL24}, but also online resource allocation
\cite{AWY14,MR12,kesselheim2014primal,AgrawalD15,GM16,Molinaro-SODA17},
load-balancing~\cite{im2024online}, network
design~\cite{meyerson2001designing}, facility
location~\cite{DBLP:conf/focs/Meyerson01,DBLP:conf/soda/KaplanNR23}
and scheduling problems
\cite{albers2020scheduling,DBLP:conf/fsttcs/AlbersJ21,AlbersGJ23}.

Online learning is a vast field, see,
e.g.,~\cite{DBLP:journals/ftml/Bubeck15,OCObook,vishnoi2021algorithms,orabona-v7,bansal2019potential};
we specifically draw on \cite{orabona-v7} for the
additive/multiplicative form of regret best suited for our
analysis. Connections between multiplicative-weights/mirror descent
and online algorithms have been used in prior works (see,
e.g.,~\cite{DBLP:journals/fttcs/BuchbinderN09,DBLP:conf/esa/BuchbinderCN14,DBLP:conf/stoc/BubeckCLLM18,DBLP:conf/soda/BuchbinderGMN19}
for problems in the adversarial arrival model, and
\cite{AgrawalD15,GM16} for random order models); however, the
connection to stochastic mirror descent we present in this work is
new, to the best of our knowledge.

\section{Online Convex Optimization}\label{sec:prelims}

Online Convex Optimization (OCO) is a fundamental model of optimization under uncertainty with fascinating applications to learning and algorithms that generalizes the classic problem of ``prediction with expert advice'' to convex loss (or concave gain) functions~\cite{DBLP:journals/toc/AroraHK12,orabona-v7}. Multiple OCO models are available, but for our purposes we focus on adversarial arrivals and stochastic first-order oracle feedback.

\begin{definition}[OCO with Stochastic First-Order Feedback]
  Let $M$ be a nonnegative real number and consider the scaled simplex
  $\scalesmpx = \{y \in [0,1]^d \mid \ip{c,y} \le M\}$. At each time, we
  need to play an action $y_t \in \scalesmpx$ with respect to an
  \emph{a priori} unknown concave function $h_t$, to maximize the long-term
  gain $\sum_t h_t(y_t)$. However, the feedback is stochastic: instead
  of getting a subgradient $\partial h_t(y_t) \in \R^{d}$ at timestep
  $t$, we get a random estimate $H_t \in \R^{d}$ such that its expectation conditioned on everything that happened up to this point (including $y_t$ and $h_t$) equals $\partial h_t(y_t)$.
\end{definition}

The key results in this area are online strategies that obtain gain comparable to the best action $y^\star \in \scalesmpx$ that knows all the gain functions in hindsight. We will make use of the following such result; because of its specifics (e.g., multiplicative/additive guarantee, scaled simplex, availability of only stochastic subgradient, etc.) we did not find this exact result in the literature, but  it follows from known regret arguments based on local norms (Appendix \ref{sec:stoc-OMD}).

\begin{restatable}[OCO Algorithm]{theorem}{thmregretConcaveScaled} \label{thm:regretConcaveScaled}
  Suppose the functions $h_t$ are concave and the stochastic gradients are ``$\ell_\infty$-bounded'': i.e.,
  $H_{t,i} \cdot (\nicefrac{M}{c_i}) \in [0,1]$ for each coordinate
  $i \in [d]$, with probability $1$. Then for any $\eta \in (0,1]$ and any stopping time $\tau$,
  the stochastic online mirror descent (OMD) algorithm ensures
  that for all $y^\star \in \scalesmpx$, we have
  \begin{align*}
    \EE\bigg[\sum_{t \le \tau} h_t(y_t)\bigg] \geq (1-\eta) \cdot \EE\bigg[\sum_{t \le \tau} h_t(y^\star)\bigg] - \frac{O(\log d)}{\eta}.
  \end{align*}
\end{restatable}
Note that the randomness is over the subgradient estimates, which in
turn affect our iterates and thus possibly the future
$h_t$'s. We will use \Cref{thm:regretConcaveScaled} with $\eta = \nf12$.

\section{Warm-Up: Random-Order (Weighted) Set Cover, Revisited}
\label{sec:warmup}

In this section, the goal is to show how to use a low-regret Online
Convex Optimization routine in a black-box fashion to obtain an optimal
competitive ratio for the online (weighted) set cover problem, when
the input is presented in random order \citep{GuptaKL21}.

\subsection{The \textsc{ROSC-OCO} Algorithm}
\label{sec:setcover-algo}

As in the introduction, we use $(\cU, \cS)$ to denote the set system of the set cover instance, with $|\cU| = n$ and $|\cS| = m$. We use $c_S$ to denote the cost of set $S \in \cS$. 

Let $\Suffix^t$ be the random subset of elements not seen until time
$t$. 
By using guess-and-double, we assume that we know an over-estimate $\optest$ for the optimal value
$c(\OPT)$, namely $c(\OPT) \le \optest \le 2c(\OPT)$; more about this in \Cref{sec:guess}. We then let
$\scalesmpx := \{p \in [0,1]^m \mid \ip{c,p} \leq
\optest\}$ be the possible vectors of sampling weights over the sets $\cS$ with expected cost at most $\optest$.

Finally, let $\kappa_e$ denote the cost of the cheapest set containing
$e$. Let $X^t_e$ be the indicator of the event that element
$e \in \cU$ is still uncovered after the first $t-1$ timesteps, and
define $\kappa_e^t := X^t_e \kappa_e$.

We then have \Cref{alg:roscoco}, following the idea described in the introduction, namely we cover (if needed) the incoming element buying one set at minimum cost $\kappa_e^t$ and then sample and buy additional sets based on the scaled sampling weights $(\nf{\kappa^t_{e^t}}{\optest}) \cdot p^t_S$; the weights $p^t$ are then updated via an OCO algorithm. While any black-box OCO algorithm can be used, for concreteness we use the one from \Cref{thm:regretConcaveScaled} with parameter $\eta = \nf12$. 

\begin{algorithm}
  \caption{\textsc{ROSC-OCO}}
  \label{alg:roscoco}
  Define $\Suffix^1 = \cU$ \\
  \ForEach{time $t = 1,2, \ldots$}{
    Obtain the next element $e^t$, which is a uniformly random element of $\Suffix^t$.\\
      
      Cover $e^t$ (if needed) by paying $\kappa^t_{e^t}$ 
      \tcp*{\textsc{Backup}} \label{step:backup-sc}

      Add each set $S \in \cS$ to the cover independently w.p.\
      $(\nf{\kappa^t_{e^t}}{\optest}) \cdot p^t_S$
      \tcp*{\textsc{Sampling}} \label{step:sample-sc}

      Feed $\gain_{t,e^t}(p) = (\nicefrac{\kappa^t_{e^t}}{\optest})\cdot \min(1,\sum_{S \ni e^t} p_S)$ to $\stococo$ to get $p^{t+1} \in
      \scalesmpx$ 
      \tcp*{\textsc{OCO}} \label{step:oco-sc}

      $\Suffix^{t+1} \gets \Suffix^t \setminus \{e^t\}$.
  }
\end{algorithm}

\subsection{Analysis of \textsc{ROSC-OCO}} 
\label{sec:setcover-analysis}

The main guarantee in this section is the following: 
\begin{theorem}\label{thm:rosc-oco}
   \textsc{ROSC-OCO} is $O(\log (mn))$-competitive for random-order (weighted) set cover.
\end{theorem}

To prove this, consider any timestep $t$, and 
let $\cH^t$ denote the history strictly before time $t$. We claim
that, for any choice of $\cH^t$, the (conditional) expected cost of
the algorithm at time $t$ is at most
$2\EE_{e \sim \Suffix^t}[X^t_e\kappa_e]$. Indeed, the element $e$
arriving at time $t$ is a uniformly random element of $\Suffix^t$, and
if it is not covered, we incur cost $\kappa_e$ in
Line \ref{step:backup-sc}, and an expected cost of at most $\kappa_e$
in Line \ref{step:sample-sc}. The randomness in the algorithm's decisions
at time $t$ is denoted as $\coins^t$ in the remainder.

To track the progress of the algorithm, define the potential as the sum of the minimal individual covering costs for the remaining elements:
\begin{gather}
  \Phi^t := \sum_{e \in \Suffix^t} \kappa^t_e;
\end{gather}
this equals $\sum_{e \in \cU} \kappa^t_e$, since all elements in
$\cU \setminus \Suffix^t$ are already covered by time $t$. Define the
stopping time $\tau$ to be the first time $t$ when
$\Phi^t \leq \optest$. By the discussion above, the expected cost
incurred at all timesteps $t \geq \tau$ is at most
$2\Phi^{\tau} \leq 2\optest$, which is at most $O(c(\OPT))$ by assumption. 

\subsubsection{The ``Real Gain'' captures the Potential Reduction}

We now bound the cost incurred in timesteps $t < \tau$. Define the
``real gain'' at time $t$ to be the expected proportional reduction in
potential due to a random element from $\Suffix^t$:
\begin{gather}
  \realgain_t := \frac{1}{|\Suffix^t|}\cdot \frac{1}{\Phi^t} \cdot
  \sum_{e, f \in \Suffix^t} \kappa^t_f \cdot \Pr[f \text{ covered by
    sampling step $t$}\mid e^t = e]. \label{eq:3}
\end{gather}

The first lemma simply uses this definition and the
inequality $1-z \leq e^{-z}$ to get:
\begin{lemma}[Bounding the Real Gain]
  \label{lem:gain-sc}
  \begin{align}\label{eq:mid-rosc}
    \EE\bigg[\sum_{t < \tau} \realgain_t\bigg] \le 1 + \E{}{ \log
    \left(\frac{\Phi^0}{\Phi^{\tau-1}}\right)} \leq O(\log n).
  \end{align}
\end{lemma}

\begin{proof}
  Consider time $t$, and condition on the history $\cH^t$ of sets
  selected thus far. The expected change in potential is
  \begin{align}
    \E{\coins^t, \; e^t\sim\Suffix^t}{\Phi^t - \Phi^{t+1} }
    &\geq \sum_{e \in \Suffix^t} \frac{1}{|\Suffix^t|}
      \sum_{f \in \Suffix^t} \kappa^t_f
      \cdot \Pr[f \text{ covered by sampling step $t$}\mid e^t =
      e] \notag \\
    &= \Phi^t \cdot \realgain_t, \notag
  \end{align}
  by the definition of $\realgain_t$. (We have an inequality since the
  potential may also drop due to the set chosen in Line \ref{step:backup-sc}.)
  Since the value of $\Phi^t$ is completely fixed by the choices in
  $\cH^t$, we get that
  \[ 
    \E{\coins^t, \; e^t\sim\Suffix^t}{\Phi^{t+1} } = \Phi^t
    \cdot \left(1 - \realgain_t\right) \leq \Phi^t \cdot
    \exp\left(-\realgain_t\right). 
  \]
  Taking logarithms and using Jensen's inequality:
  \[ 
    \E{\coins^t, \; e^t\sim\Suffix^t}{\log (\Phi^{t+1})} \leq
    \log \Big(\E{\coins^t, \; e\sim\Suffix^t}{\Phi^{t+1}}\Big) \leq \log (\Phi^t)  - \realgain_t.
  \]
  Since this is true for each conditioning $\cH^t$, we get that
  $\E{}{\log (\Phi^{t+1})} \leq \E{}{\log (\Phi^{t})} - \E{}{\realgain_t(p^t)}$.
  Finally, summing over all times up until the stopping time $\tau -2$,
  we obtain
    \begin{align*}
    \EE\bigg[\sum_{t < \tau-1} \realgain_t(p^t)\bigg] \le \E{}{ \log \left(\frac{\Phi^0}{\Phi^{\tau-1}}\right)} \le O(\log n),
  \end{align*}
  where in the last inequality we used that
  $\Phi^{\tau-1} > \optest \ge c(\OPT)$ by definition of $\tau$ and
  $\optest$ and
  $\Phi^0 = \sum_{e \in \cU} \kappa_e \le n \cdot
  c(\OPT)$.
  Finally, $\realgain_t \leq 1$ for each time $t$; using this for $t =
  \tau-1$ completes the proof. 
\end{proof}  

\subsubsection{Simplifying the ``Real Gain''}

The ``real gain'' function is difficult for us to work with, so we
define a ``decoupled'' version where we essentially linearize the probability of covering element $f$ during the sampling step as $\frac{\kappa^t_f}{\optest} \cdot \sum_{S: f \in S} p_S$:
\begin{align}
  \gain_{t}(p) := \E{f \sim \Suffix^t}{\gain_{t,f}(p)} \qquad \text{where} \qquad
  \gain_{t,f}(p) := \frac{\kappa^t_f}{\optest} \cdot \min\Big(1,
  \sum_{S: f \in S} p_S\Big). \label{eq:gain-sc}
\end{align}
Note that the gain function $\gain_t$ is an expectation over
$\gain_{t,f}$ functions, and hence convenient to use in a stochastic
gradient descent subroutine; indeed, this is precisely what we used in
Line \ref{step:oco-sc} of the algorithm.

\begin{lemma}[Gain vs.\ Real Gain]
  \label{lem:gain-realgain}
  $\gain_t(p^t) \leq \frac{e}{e-1} \cdot \realgain_t$. 
\end{lemma}

\begin{proof}
  Using the sampling probability in Line \ref{step:sample-sc} of the
  algorithm in the definition~(\ref{eq:3}),
  \begin{align*}
    \realgain_t 
    &= \frac{1}{|\Suffix^t| \cdot \Phi^t}
      \sum_{e, f \in \Suffix^t} \kappa^t_f
      \cdot \bigg( 1 - \prod_{S: f \in S} \bigg(1 - \frac{p^t_S
      \kappa^t_e }{\optest}\bigg)\bigg).
  \end{align*}
  Define $a_S = \frac{p_S^t \kappa_e^t}{\optest}$; we can use that 
  $1 - \prod_{S} (1-a_S) \geq (1-\nf1e) \min(1, \sum_{S} a_S)$ (see
  \Cref{clm:1-1overe}) to get
  \begin{align}
    \realgain_t 
    &\geq  (1-\nf1e) \cdot \frac{1}{|\Suffix^t|\cdot \Phi^t}
      \sum_{e,f \in \Suffix^t} \kappa^t_f 
      \cdot \min\Big(1, \sum_{S: f \in S} \frac{p^t_S
      \kappa^t_e }{\optest}\Big) \label{eq:1} \\
    &\geq  (1-\nf1e) \cdot \frac{1}{|\Suffix^t|\cdot \Phi^t}
      \sum_{e,f \in \Suffix^t} \kappa^t_f 
      \cdot \frac{\kappa^t_e }{\optest} \cdot \min\Big(1, \sum_{S: f
      \in S} p^t_S \Big). \label{eq:2}
  \end{align}
  where we used that $\kappa_e^t \leq c(\OPT) \leq \optest$ to
  get~(\ref{eq:2}).  Now using that
  $\Phi^t = \sum_{e \in \Suffix^t} \kappa^t_e$, and the definition of
  $\gain_t$ shows that $\realgain_t \geq (1 - \nf1e) \cdot \gain_t(p^t)$
  and completes the proof.
\end{proof}

\subsection{Sufficient Static Gains}

\Cref{lem:gain-sc,lem:gain-realgain} bounded the sum of our gains over all timesteps by $O(\log n)$. We now show the algorithm's expected cost can be upper bounded not by our gains, but by the gains of the optimal (fractional) solution and the optimal cost.

\begin{lemma}\label{lemma:suffSG}
 If $p^\star$ denotes the optimal
    (fractional) solution of the linear programming
    relaxation for the set cover problem, then
    \begin{align*}
    \EE\bigg[\sum_{t < \tau} c(\ALG^t)\bigg] \le 4c(\OPT) \cdot \EE\bigg[\sum_{t < \tau}  \gain_t(p^\star)\bigg].
    \end{align*}    
\end{lemma}

\begin{proof}
    . By definition of $\gain_{t,f}$ we obtain
    \begin{align*}
      \gain_{t,f}(p^\star)
      &= \frac{\kappa^t_f}{\optest}
        \cdot \min \Big(1, \sum_{S \in \cS : S \ni f} p^\star_S \Big) = \frac{\kappa^t_f}{\optest},
    \end{align*}
    where we used that $p^\star$ is a fractional set cover and hence
    covers $f$ at to an extent of at least $1$. Now substituting
    into~(\ref{eq:gain-sc}),
    \begin{align*}
        \gain_t(p^\star)
         &= \frac{1}{\optest} \cdot \underset{f \sim \Suffix^t}{\EE}[\kappa^t_f]
         \ge \frac{{\EE}[c(\ALG^t)]}{4c(\OPT)},
    \end{align*}
    the inequality using that the algorithm's expected cost is at
    most $2\EE_{f \sim \Suffix^t}[\kappa^t_f]$ and
    $\optest \le 2c(\OPT)$. Rearranging, adding up to time $\tau - 1$, and taking expectations gives the claimed bound.
\end{proof}

\subsubsection{Low Regret using Stochastic OCO Implies Competitiveness}

The final step of the analysis is to use the low-regret property of the online convex optimization procedures. This allows us to relate the algorithm’s cumulative gain to that of the optimal solution $p^\star$, thereby completing the argument that upper-bounds the algorithm’s expected cost. To this end, we first establish the following properties of the gain function.
algorithm.
To proceed, we need the following properties of the gain
function.

\begin{claim}[Gain is Well-Behaved]
  \label{clm:gain-good-sc}
  For each $t$ and each $e \in \Suffix^t$, the function
  $\gain_{t,e}(\cdot)$ is non-negative, concave, and moreover each coordinate
  $(\partial \gain_{t,e})_S$ of its subgradient lies in $[0,c_S/\optest]$.
\end{claim}

\begin{proof}
  Non-negativity follows from the definition, and concavity follows from concavity of the function $x \mapsto \min(1,\sum_i x_i)$; for the second, note
  that
  \begin{align*}
    \frac{\partial}{\partial p_S} \gain_{t,e}(p) \cdot
    \frac{\optest}{c_S} \leq \frac{\kappa^t_e}{c_S} \cdot \ind{e
    \in S}\le 1,
  \end{align*}
  where we use that $\kappa^t_e \leq \kappa_e \leq c_S$ for any set
  $S$ that contains $e$.
\end{proof}

We are finally in a position to complete the proof of \Cref{thm:rosc-oco}.
\begin{proof}[Proof of \Cref{thm:rosc-oco}]
 The subgradient $\partial \gain_{t,e^t}$ fed to the \stococo procedure in Line \ref{step:oco-sc} is an unbiased estimate for the subgradient $\partial \gain_t$ of the gain function. Since \Cref{clm:gain-good-sc} also ensures that the conditions of \Cref{thm:regretConcaveScaled} are satisfied, it guarantees that the procedure computes $p^t$'s with gains competitive to those of the optimal fractional set cover solution $p^\star$, namely (recall we used $\eta = \nf12$)
    \[
      \EE\bigg[\sum_{t < \tau} \gain_t(p^t)\bigg] \geq \frac{1}{2} \cdot \EE\bigg[\sum_{t <
      \tau} \gain_t(p^\star)\bigg] - O(\log m),
    \]
    This connects \Cref{lem:gain-sc,lem:gain-realgain} (which bound the sum of our gains by $O(\log n)$) to \Cref{lemma:suffSG} to give 
    \begin{align*}
        \EE\Big[\sum_{t < \tau}c(\ALG^t)\Big] \le O(\log(mn)) \cdot
        c(\OPT). 
    \end{align*}
    Combining this with the expected cost the algorithm incurs after stopping time $\tau$ we obtain
    \[
        \E{}{c(\ALG)} =  \EE\Big[\sum_{t < \tau}c(\ALG^t)\Big] +
        \EE\Big[\sum_{t \ge \tau} c(\ALG^t)\Big]  \le O(\log(mn)) \cdot
        c(\OPT), 
    \]
    as desired.
\end{proof}

\subsection{The Guess-and-Double Framework} \label{sec:guess}

We have assumed that we know an estimate $\optest$ of the optimal value satisfying $c(\OPT) \le \optest \le 2 c(\OPT)$. If computational complexity is not a concern, this can be accomplished by using the guess-and-double framework, namely starting with the estimate $\optest$ being the cost of the cheapest set, and doubling $\optest$ and restarting the algorithm whenever the (integral) optimum of the instance seen thus far exceeds $\optest$. 

In order to obtain a polynomial-time algorithm, one can replace
$c(\OPT)$ throughout for the  use the cost of the linear programming
relaxation $c(\LP)$ (i.e., $c(\LP) := \min\{\sum_{S \in \cS} c_S x_S :
\sum_{S : e \in S} x_S \ge 1,~\forall e \in \cU\}$) instead of
$c(\OPT)$. Notice that indeed the solution $p^{\star}$ used in the
proof of \Cref{lemma:suffSG} was already the fractional optimal
solution. The only other place where any property of the optimal
solution is used is in upper bounding the (log of) the initial
potential $\Phi^0$ at the end of \Cref{lem:gain-sc}, but it still
holds that $\Phi^0 = \sum_{e \in \cU} \kappa_e \le n \cdot c(\LP)$,
since even the fractional optimum is at least the minimum cost set
containing each element.
Moreover, this bound becomes
$\Phi^0 = \sum_{e \in \cU} \kappa_e \le n \cdot \textsc{IPGap} \cdot
c(\LP)$, where $\textsc{IPGap} = \nicefrac{c(\OPT)}{c(\LP)}$; as long as this is at most polynomial in $n$, this only changes
constant factors in the bound of
$\log (\nicefrac{\Phi^0}{\Phi^{\tau - 1}})$. Since the fractional optimum
for the instance seen thus far can be computed in polytime, one can
then use the above guess-and-double strategy and be able to assume the
estimate $c(\LP) \le \optest \le 2 c(\LP)$.

\medskip

In the following section, we abstract out aspects of the set cover
analysis to get a general recipe for the class of monotone covering
problems.

\section{OCO for Random-Order Monotone Problems}\label{sec:framework}

We now give our framework for monotone covering problems in the
random order setting. Our goal will be to abstract out the generic
steps from the problem-specific ones, which can help focus on the
essential core of these problems.

\subsection{General Problem Formulation}

We study the following class of covering problems, which are online
minimization problems characterized by (i)~a (finite) set of
\textit{resources} $\cS$ (with $|\cS|=m$), and (ii)~a (finite) set of
\textit{elements} $\cU$ (with $|\cU|=n$). For set cover, the resources
are subsets of elements in some given family, which can be used to
cover elements they contain, and for facility location the resources are the possible facilities.

Each element needs to be ``satisfied'' or ``covered'', and for that we
need to buy resources and pay the augmentation cost. More precisely,
each resource $R \in \cS$ has a cost $c_R$. For every element
$e \in \cU$ and (multi)set of resources $\cR$, the augmentation cost
incurred if resources $\cR$ have already been bought is denoted by
$\aug(e, \cR) \geq 0$. For example, for set cover $\aug(e, \cR)$ is
$0$ if $e$ belongs to some set in $\cR$, otherwise it is the cost of
the cheapest set containing $e$. In facility location, $\aug(e, \cR)$
is the minimum between connecting element/demand $e$ to a facility in
$\cR$ or opening and connecting to a new facility outside of $\cR$; we
give more details in \Cref{sec:nmfl}.

We assume the following \emph{monotonicity} and \emph{minimality} properties of the augmentation cost; the first says that the fewer resources we have, the higher the augmentation cost is, and the second says that it is possible to ``buy new resources to reduce the augmentation cost'':  

\begin{enumerate}
    \item (monotonicity) $\aug(e, \cR) \geq \aug(e, \cR')$ if $\cR \sse \cR'$.

    \item (minimality) $\aug(e,\cR) \le \aug(e,\cR') + \sum_{R \in \cR' \setminus \cR} c_R$ if $\cR \sse \cR'$.
\end{enumerate}

The online setting is then: the costs $c$ and $\aug(\cdot,\cdot)$ are
given upfront, and items $e^1,e^2,\ldots,e^n \in \cU$ arrive
one-by-one. At time $t$, we assume that the
  multiset of resources previously bought is denoted by $\cR^t$. When element $e^t$ arrives, the algorithm does the following:
\begin{enumerate}[label=(\roman*)]
\item \label{item:buy} it decides which (if any) additional resources it wants to buy,
  updating the (multi)set of obtained resources to
  $\cR^{t+1} \supseteq \cR^{t}$ and incurring  cost $\sum_{R \in
    \cR^{t+1}\setminus \cR^t} c_R$, and 
\item \label{item:pay} it also incurs the additional augmentation cost
  $\aug(e^t, \cR^{t+1})$. Often we consider the larger quantity 
    $\aug(e^t, \cR^t)$ instead, since the two differ by at most the cost
    incurred in~\ref{item:buy}  due to the minimality property above.
\end{enumerate}
The goal is to minimize the total amount spent on buying the resources
plus the per-step augmentation costs, namely
$\sum_{R \in \cR^n} c_R + \sum_{t \le n} \aug(e^t, \cR^{t+1})$. As always,
$\OPT$ denotes the set of resources
of the offline optimal solution, and $c(\OPT)$ its cost. We note that the ``minimality''
property of the augmentation cost implies that the initial
augmentation cost $\aug(e, \emptyset)$ of an element $e$ is upper
bounded by $\OPT$, as
$\aug(e, \emptyset) \le \sum_{R \in \OPT} c_R + \aug(e, \OPT) \le
c(\OPT)$.

\subsection{The General Framework}

The \textsc{ROSC-OCO} algorithm for random order (weighted) set cover
can be extended to a general framework for more general covering
problems. As before, let $\Suffix^t$ be the elements which have not
been seen in the first $t-1$ timesteps, so that $\Suffix^1 = \cU$.
Let $\cR^t \sse \cS$ be the resources that have been chosen by the
algorithm by the beginning of the $t^{th}$ timestep, and hence
$\cR^1 = \emptyset$. Again we assume we know an estimate $\optest$ such that
$c(\OPT) \le \optest \le 2c(\OPT)$, following the discussion from \Cref{sec:guess}; as in the previous section, let
$\scalesmpx := \{p \in [0,1]^m \mid \ip{c,p} \leq \optest\}$.

In our framework, we define an \emph{augmentation cost estimate}
$\kappa_e^t$, with the property that
\begin{gather}
  \aug(e, \cR^t)/\approxfactor \leq \kappa_e^t \leq \aug(e, \cR^t) \label{eq:kappa-gen}
\end{gather}
for some parameter $\approxfactor \geq 1$. We require that $\kappa_e^t$
should be efficiently computable given $e$ and $\cR^t$, and that we
have access to an augmentation algorithm $\AUG$ that, given any $e$
and $\cR^t$, can output an augmentation set $\cA^t$ with cost at most
$\approxfactor \kappa_e^t$. Our algorithm, given in~\Cref{alg:augoco}, is very similar to the one
for set cover. Akin to set cover, the randomness in the algorithm's decisions at time $t$ is denoted by $\coins^t$.

\begin{algorithm}
  \caption{\augoco}
  \label{alg:augoco}
  Define $\Suffix^1 = \cU$ \\
  \ForEach{time $t = 1,2, \ldots$}{
    Let $e^t$ be the next element, which is a uniformly random element of $\Suffix^t$.\\
      
      Cover $e^t$ by using augmentation algorithm $\AUG$ paying $\le \approxfactor \kappa^t_{e^t}$ 
      \tcp*{\textsc{Backup}} \label{step:backup-gen}

      Buy each resource $S \in \cS$ independently w.p.\
      $(\nf{\kappa^t_{e^t}}{\optest}) \cdot p^t_S$
      \tcp*{\textsc{Sampling}} \label{step:sample-gen} 

      Feed a suitable function $\gain_{t,e^t}$ to $\stococo$ to get $p^{t+1} \in
      \scalesmpx$ 
      \tcp*{\textsc{OCO}} \label{step:oco-gen} 
      
      $\Suffix^{t+1} \gets \Suffix^t \setminus \{e^t\}$.
  }
\end{algorithm}

Line~\ref{step:oco-gen} is not fully specified---the gain function
$\gain_{t,e}$ will be problem-specific. In the next section, we
present the properties we require from this gain function to enable
our analysis. Again, for concreteness, throughout we use in Line~\ref{step:oco-gen} the OCO algorithm from \Cref{thm:regretConcaveScaled} with $\eta = \nf{1}{2}$.

\subsection{The Analysis}
\label{sec:gen-analysis}

Define the potential
function
\begin{gather}
  \Phi^t = \sum_{e \in \Suffix^t} \kappa_e^t,  %
  \label{eq:phi-gen}
\end{gather}
which is, up to the parameter $\alpha$, the total augmentation costs $\aug(e,\cR^t)$ of future items $e \in \Suffix^t$ given the current resources $\cR^t$. The analysis framework is again similar to that for set cover:

\begin{enumerate}
\item The expected cost at any timestep $t$ is at most
  $(\approxfactor + 1)\, \kappa_e^t$, since we spend at most
  $\approxfactor \kappa_e^t$ in Line~\ref{step:backup-gen}, and at most
  $\kappa_e^t$ in expectation in Line~\ref{step:sample-gen}.

\item We define the stopping time $\tau$ to be the first timestep for
  which $\Phi^\tau \le \optest$. Consider the elements $e^{\tau},
  \ldots, e^n$ in the order they appear after time $\tau$: the
  expected cost incurred is at most
  \[ \sum_{t \geq \tau} (\approxfactor + 1) \cdot \kappa_{e^t}^t \leq \sum_{t \geq \tau} (\approxfactor + 1) \cdot \kappa_{e^t}^\tau  = (\approxfactor+1) \cdot \Phi^\tau.\]
  This is at most $2(\approxfactor+1) c(\OPT)$ using the definition of $\tau$. Hence we
  now bound the cost incurred prior to the stopping time.
  
\item Consider timestep $t$, and condition on the history $\cH^t$ of
  everything that happened before timestep $t$. Now define the real
  gain function to be
  \begin{gather}
    \realgain_t := \frac{1}{|\Suffix^t|}\cdot \frac{1}{\Phi^t} \cdot \sum_{e, f \in \Suffix^t} \underset{\coins^t}{\EE}[ \kappa^t_f - \kappa^{t+1}_f \mid e^t = e], \label{eq:4-gen}
  \end{gather}
  where $\kappa^{t+1}_f$ depends on the sampled set of resources chosen in Line~\ref{step:sample-gen} of the algorithm.

\item The expected change in potential at timestep is captured by the real gain:
  \begin{align}
    \E{\coins^t, \; e\sim\Suffix^t}{\Phi^t - \Phi^{t+1} }
    &\geq \sum_{e \in \Suffix^t} \frac{1}{|\Suffix^t|}
      \sum_{f \in \Suffix^t} \underset{\coins^t}{\EE}[ \kappa^t_f - \kappa^{t+1}_f \mid e^t = e] \notag \\
    &= \Phi^t \cdot \realgain_t, \notag
  \end{align}
  by the definition of $\realgain_t$. Then, the same analysis as in
  \Cref{sec:setcover-analysis} shows that 
  \begin{align}
    \EE\bigg[\sum_{t < \tau} \realgain_t\bigg] \le 1 + \E{}{ \log
    \left(\frac{\Phi^0}{\Phi^{\tau-1}}\right)} \le O(\log n), \label{eq:le-tau}
  \end{align}
  where again the last inequality uses that
  $\Phi^{\tau-1} > \optest \ge c(\OPT)$ by definition of $\tau$ and
  $\optest$ and
  $\Phi^0 = \sum_{e \in \cU} \kappa^0_e \le  \sum_{e \in \cU} \aug(e,\emptyset) \le n \cdot c(\OPT)$.
\end{enumerate}

Thus far, the analysis has been completely generic, and applies to any
monotone augmentation cost. We finally need some problem-dependent
properties to upper bound the cost of the algorithm by this sum of real gains $\sum_t \realgain_t$, connecting to the $\gain_{t,e^t}$ employed in the algorithm. The following lemma captures the properties we need to
turn a low-regret guarantee for the $\stococo$ subroutine into one
showing bounded competitive ratio.

\begin{lemma}\label{lem:gain-oco-application}
  Assume that for each $t$, there exists a
  gain function $\gain_t(p): \scalesmpx \rightarrow \R_+$ such that:
  \begin{enumerate}[label=(A\arabic*)]
  \item \label{item:ass1} \emph{(Lower Bound on Real Gain)} There exists a constant
    $\gamma \in (0,1]$ such that for all $t \in [n]$, we have
    \[ \gamma\, \gain_t(p^t) \leq \realgain_t; \]
  \item \label{item:ass2} \emph{(Unbiasedness)} the gain function can
    be written as
    \[ \gain_t(p) = \underset{e \sim \Suffix^t}{\EE} [\gain_{t,e}(p)], \]  and hence
    the subgradient $G_t = \partial \gain_{t,e^t}(p^t)$ gives an
    unbiased estimate of $\partial \gain_{t}(p^t)$;
  \item \label{item:ass3} \emph{(Bounded Gradients)} These unbiased estimates
    $G_t$ satisfy $(G_t)_i \in [0, c_i/\optest]$, and 
  \item \label{item:ass4} \emph{(Sufficient Static Gains)} There exists a vector
    $p^\star \in \scalesmpx$ and some $\beta, \delta > 0$ such that
    \[
      \EE\bigg[\sum_{t < \tau} c(\ALG_t)\bigg] \leq c(\OPT) \cdot
      \bigg( \beta \; \EE\bigg[\sum_{t < \tau} \gain_t(p^\star) \bigg]
      + \delta \bigg).
    \]
  \end{enumerate}
  Then, it holds that the cost $c(\ALG)$ of \Cref{alg:roscoco} is bounded by
  \[
    \E{}{c(\ALG)} \leq \Big( O(\approxfactor) + \frac{O(\beta\, \log
      mn)}{\gamma} + \delta \Big) \cdot c(\OPT).
  \]
\end{lemma}
\begin{proof}
  Since $G_t = \partial \gain_{t,e^t}(p)$ is an unbiased estimate of
  $\partial\gain_t(p^t)$ by Assumption~\ref{item:ass2}, and
  $G_{t,i} \in [0, c_i/\optest]$ by Assumption~\ref{item:ass3}, we can
  use \Cref{thm:regretConcaveScaled} with $\eta = \nf{1}{2}$ to infer that for any vector
  $p^\star \in \scalesmpx$,
  \[
    \EE\bigg[ \sum_{t < \tau} \gain_t(p^t) \bigg] \geq (1-\eta) \cdot
    \EE\bigg[\sum_{t < \tau} \gain_t(p^\star) \bigg]- \frac{O(\log
      m)}{\eta} = \frac 12 \cdot \EE\bigg[ \sum_{t < \tau} \gain_t(p^\star)
    \bigg] - O(\log m).
  \]
  Moreover, combining Assumption~\ref{item:ass1} with (\ref{eq:le-tau}):
  \[
    \EE\bigg[\sum_{t < \tau} \gain_t(p^t)\bigg] \le \frac1\gamma \cdot
    \EE\bigg[\sum_{t < \tau} \realgain_t\bigg] \le \nf1\gamma \cdot O(\log n).
  \]
  Combining the two inequalities above with Assumption~\ref{item:ass4} gives
  \[
    \sum_{t < \tau} \E{}{c(\ALG_t)} \le c(\OPT) \cdot \bigg( \frac{O(\beta \, \log
      mn)}{\gamma} + \delta \bigg). 
  \]
  Finally, adding in the $O(\approxfactor\cdot c(\OPT))$ cost after the
  stopping time completes the proof.
\end{proof}

To summarize, using this framework to get $O(\log mn)$-competitive
algorithms requires the following:
\begin{enumerate}[label=(\roman*)]
\item Define a suitable proxy $\kappa_e^t$ which is an under-estimate
  for the augmentation cost $\aug(e, \cR^t)$, and also show an
  augmentation algorithm $\AUG$ that covers $e$ with cost at most
  $\approxfactor\, \kappa_e^t$. 
\item Give a ``decoupled'' function
  $\gain_t = \mathbb{E}_{e \in \Suffix^t} [\gain_{t,e}]$ which satisfies the
  four assumptions in~\Cref{lem:gain-oco-application}, for values
  $\gamma, \beta \in \Theta(1)$. Recall this gain function is the one
  we will use in Line~\ref{step:oco-gen} of \Cref{alg:augoco}.

\end{enumerate}

\section{Applications}\label{sec:applications}

In this section, we apply the framework presented in
\Cref{sec:framework} to three problems considered in \citep{GuptaKL21,
  GuptaKL24}. For each of these applications, we consistently follow a
two-part analysis: first, we employ a version of $\augoco$
(\Cref{alg:augoco}) that is specifically instantiated for the
considered application. Specifically, for every application, we
specify what the augmentation cost is, since this is the feature differentiating the algorithms. Second, we verify the conditions outlined in \Cref{lem:gain-oco-application} to prove a tight competitive ratio guarantee for our algorithm.

\subsection{Unweighted Set Multicover}\label{sec:smc}

In the random-order (unweighted) set multicover problem, the resources
$\cS$ are again, but now each element $e \in \cU$ arrives with an
integer demand $b_e \in \Z_+$. We seek a solution to the following
IP:
\begin{align}
    &\min \;\ip{\bs{1}, x} \notag\\
    &~\text{s.t.}~ Ax \geq b \tag{\textsc{SMC}} \label{eq:rosmc}\\
    &\qquad x \in \{0,1\}^m, \notag
\end{align}
where $A \in \{0,1\}^{n \times m}$ is the set-element incidence
matrix, and $b \in \Z^n_+$ is the vector of demands. Recall that $\cR^t$ denotes the resources/sets picked by the algorithm before time $t$. At any time $t$,
let $b^t_e := \max(0, b_e - \sum_{S \in \cR^t} A_{eS} x^t_S)$ denote the ``uncoverage''
for element $e \in \cU$ at time $t$; we ensure that all elements
$e \in \cU \setminus \Suffix^t$ have $b^t_e = 0$---i.e., elements are
covered once they arrive.

\subsubsection{Augmentation Costs and the Potential}

Since we have unit set costs, we can pick $b^t_e$
as-yet-unpicked sets to cover $e$, and hence we have
$\aug(e, \cR^t) := b^t_e$. Define the augmentation cost estimate as the cost itself, namely $\kappa^t_{e} = b^t_{e}$, so that
$\approxfactor = 1$ in (\ref{eq:kappa-gen}). Conditioning on the
history $\cH^t$ until time $t$, the real gain at time $t$ is
\begin{align}
  \realgain_t &= \frac{1}{\Phi^t \cdot |\Suffix^t|} \sum_{e,f
                \in \Suffix^t} \underset{\omega^t}{\EE}\big[b^t_f - b^{t+1}_f ~\big|~ e^t = e \big]. \label{eq:rg-smc}
\end{align}
Define its linearization/decoupling as:
\begin{gather}
  \gain_t(p) = \underset{f \sim \Suffix^t}{\EE}[\gain_{t,f}(p)] \qquad
  \text{where} \qquad \gain_{t,f}(p) = \frac{1}{\optest} \min\Big(b_f^t, \sum_{S \not\in \cR^t} A_{f S} \; p_S\Big).
\end{gather}
Using this definition of $\kappa_e^t$ and the gain function, we can
now use the framework $\augoco$; since we have unit costs, we are
working in the polytope
$\{ p \in [0,1]^n \mid \ip{\bs{1}, p} \leq \optest \}$. The
following theorem proves the competitiveness of the resulting
algorithm.

\subsubsection{Competitiveness}

\begin{theorem}\label{thm:rosmc-oco}
  There exists an $O(\log (mn))$-competitive algorithm for
  random-order (unweighted) set multicover problem.
\end{theorem}

\begin{proof}
  Fix time $t$, and condition on the history $\cH^t$ until time $t$.
  By definition, the function $\gain_t$ satisfies the unbiasedness
  property \ref{item:ass2}. We now check 
  the other properties in \Cref{lem:gain-oco-application}.

  (\emph{Lower Bound on Real Gain}): We simplify the inner expression
  in the summation in (\ref{eq:rg-smc}), further conditioning on $e^t$ (so only the randomness $\omega^t$ of the algorithm is free):
  \begin{align*}
    \E{\omega^t}{b^t_f - b^{t+1}_f ~\Big|~ e^t = e}
    &= \underset{\omega^t}{\EE}\Big[\min\Big(b_f^t, \sum_{S \not\in \cR^t} A_{f S} \cdot \ind{S \text{ is sampled at time } t} \Big)~\Big|~ e^t = e\Big]\\
    &\ge \nf{1}{168} \cdot \min\Big(b_f^t, \sum_{S \not\in \cR^t} A_{f S} \cdot \Pr[S \text{ is sampled at time } t \mid e^t = e]  \Big)\\
    &\geq \nf{1}{168} \cdot  \frac{b^t_e}{\optest} \cdot
      \min\Big(b_f^t, \sum_{S \not\in \cR^t} A_{f S} \cdot p_S^t\Big)
      = \nf{1}{168} \cdot b^t_e \cdot \gain_{t,f}(p^t),
  \end{align*}
  where the first inequality uses \Cref{fct:crs} (and the fact that
  $b^t_f \geq 1$), and the second
  inequality substitutes the algorithm's sampling probability
  $p_S^t b_e^t/\optest$ and uses that $b_e^t \leq \optest$. Substituting into~(\ref{eq:rg-smc}), we get
  \begin{align*}
    \realgain_t \geq \nf{1}{168} \cdot \gain_t(p^t).
    \end{align*}
    This proves property~\ref{item:ass1} with $\gamma = \nf1{168}$.

    (\emph{Bounded Gradients}): We have that $\nf{\partial}{\partial p_S} \; \gain_{t,f}(p) = 0$ if $b^t_f$ achieves the minimum in the $\gain_{t,f}(p)$ expression, or if the set $S$ has already been chosen before. Otherwise, 
    \begin{align*}
      \frac{\partial}{\partial p_S} \gain_{t,f}(p)
      = \frac{1}{\optest}
      A_{fS} \leq \frac{1}{\optest} = \frac{c_S}{\optest},
    \end{align*}
    where we used the unweighted setting to infer $c_S = 1$,  which shows property~\ref{item:ass3}.

    (\emph{Sufficient Static Gains}): To show this property, set
    $p^\star$ be the optimal fractional solution to the linear
    programming relaxation of~(\ref{eq:rosmc}). Then
    \begin{align}
      \gain_{t,f}(p^\star)
      &= \frac{1}{\optest} \min\Big(b_f^t,
        \sum_{S \not\in \cR^t} A_{f S} \; p^\star_S\Big) =
        \frac{b_f^t}{\optest} = \frac{\kappa_f^t}{\optest}, \label{eq:4}
    \end{align}
    where we used the feasible of $p^\star$ to infer that it must
    (fractionally) pick at least $b_f^t$ sets from $\cS \setminus \cR^t$ containing $f$. Using
    (\ref{eq:4}), that, as before, the algorithm's expected cost in Lines \ref{step:backup-gen} and \ref{step:sample-gen} is at most $\mathbb{E}_{f \in \cU^t}[\kappa_f^t] + \mathbb{E}_{f \in \cU^t}[\kappa_f^t]$ and $\optest \leq 2c(\OPT)$ in the definition of
    $\gain_t$, we get
    \[ \gain_t(p^\star) = \frac{1}{\optest} \cdot \underset{f \sim \Suffix^t}{\EE}[\kappa_f^t]
    \geq \frac{\EE[c(\ALG_t)]}{4c(\OPT)}, \] thereby satisfying
    the final property~\ref{item:ass4}. The competitive ratio of the algorithm from \Cref{thm:rosmc-oco} now follows from \Cref{lem:gain-oco-application}.    
\end{proof}

\subsection{Covering Integer Programs}\label{sec:cip}

Our second application is to random-order covering integer programs
(CIPs): each element $e$ requires a coverage of $1$, and each
resource/set provides $A_{eS} \in [0,1]$ amount of coverage for
it. Our goal is to find the cheapest set of resources that cover all
elements. In the random order model, the instance is chosen by the
adversary, and then the elements arrive in random order. Formally, we
want to solve the following IP:
\begin{align}
    &\min \;\ip{c, x} \notag\\
    &~\text{s.t.}~ Ax \geq 1 \tag{\textsc{CIP}} \label{eq:ROCIP}\\
    &\qquad x \in \Z^m_+, \notag
\end{align}
where $A \in [0,1]^{n \times m}$; the restriction of $A$ to $[0,1]$ is
without loss of generality, since the RHS is $1$. Note that we do not
have upper bounds on the variables $x$, and hence this does not
generalize set multicover. Indeed, we can pick multiple copies of each
resource, and hence, instead of $\cR^t$, we now use $x^t \in \Z^m_+$ to indicate how many copies of each set the algorithm has picked before time $t$. To reduce notation, let
$\ip{A_e, x} = \sum_{S \in \cS} A_{eS} x_S$ denote the coverage of the
resource $e$ by solution $x$.

\subsubsection{Half-Covering and the Augmentation Function}
\label{sec:half-cover-augm}

It will be convenient to consider a ``half-coverage'' variant of the
problem: we will consider the element $e$ covered by a solution $x \in \Z_+^m$
if the coverage $\ip{A_e, x}$ is at least $1/2$.  Since we do not have
any upper bounds, we can always transform such a ``half-solution'' to
a full solution by taking the solution $2x$ (i.e., by actually picking
two units of the resource when the algorithm picks one), at the
expense of doubling the cost. Henceforth, we focus on this
half-coverage variant. For an element $e$, we define
\[ b_e^t := 
  \begin{cases}
    1 - \ip{A_e, x^t} & \text{if $\ip{A_e, x^t} \leq \nf12$} \\
    0 & \text{otherwise}.
  \end{cases}
\]
Observe that $b_e^t \in \{0\} \cup [\nf12, 1]$. Next, define
\begin{gather}
  \aug(e, x^t) := \min\{ y \in \Z_+^m \mid \ip{A_e, y} \geq b_e^t
  \}. \label{eq:5}
\end{gather}
Let $S(e) := \arg\min_{S \in \cS} \frac{c_S}{A_{eS}}$ be the set that
maximizes the bang-per-buck for element $e$, and
let the optimal density be $\density_e = c_{S(e)}/A_{e S(e)}$. Define the augmentation cost estimate
\[  \kappa_e^t =
  \begin{cases}
    b_e^t \cdot \density_e & \text{if $b_e^t
                                             \geq \nf12$} \\
    0 & \text{otherwise}.
  \end{cases}
\]
The (fractional) solution which chooses $b_e^t/A_{eS(e)}$ copies of
set $S(e)$ is the optimal solution to the fractional relaxation of
(\ref{eq:5}) with cost $\kappa_e^t$; this shows that
$\kappa_e^t \leq \aug(e,x^t)$. Moreover, choosing
$\lceil b_e^t/A_{eS(e)} \rceil$ copies of set $S(e)$ is an integer
solution, and since $b_e^t \geq \nf12$ and $A_{eS} \leq 1$, we get
that $\aug(e,x^t)/4 \leq \kappa_e^t$, giving us $\approxfactor = 4$ in
the inequality~(\ref{eq:kappa-gen}). 

For a generic element $f \in \Suffix^t$, it holds that 
\[
  \kappa^t_f - \kappa^{t+1}_f = \density_f \cdot (b^t_f - b^{t+1}_f) 
\]
The real gain function at time $t$ is
\begin{align*}
  \realgain_t(p^t) &= \frac{1}{\Phi^t \cdot |\Suffix^t|} \sum_{e,f \in \Suffix^t} \underset{\coins^t}{\EE}\big[\kappa^t_f - \kappa^{t+1}_f  ~\big|~ e^t = e \big] = \frac{1}{\Phi^t \cdot |\Suffix^t|} \sum_{e,f \in \Suffix^t} \density_f \cdot \underset{\coins^t}{\EE}\big[b^t_f - b^{t+1}_f  ~\big|~ e^t = e \big].
\end{align*}
We now define the decoupled gain function:
\begin{gather}
  \gain_t(p) = \underset{f \sim \Suffix^t}{\EE}[\,\gain_{t,f}(p)\,] \qquad
  \text{where} \qquad \gain_{t,f}(p) = \frac{\rho_f}{\optest} \cdot
  \min\Big(b_f^t, \ip{A_f,p}\Big). \label{eq:6}
\end{gather}
Using this definition of
$\kappa_e^t$ and $\gain_t(p)$, we instantiate the framework $\augoco$
on the polytope
$\{ p \in [0,1]^n \mid \ip{c, p} \leq \optest \}$.

\subsubsection{Competitiveness} 

\begin{theorem}\label{thm:rocip-oco}
  There exists an $O(\log (mn))$-competitive algorithm for
  random-order CIPs.
\end{theorem}

\begin{proof}
  It suffices to verify the properties
  in~\Cref{lem:gain-oco-application}.  The unbiasedness
  condition~\ref{item:ass2} is satisfied by definition. For the
  others, fix a time $t$, and condition on the history $\cH^t$ until
  time $t$.
  
  (\emph{Lower Bound on Real Gain}): Using \Cref{fct:crs} again and
  $\kappa_e^t \leq \optest$, we get that
  \begin{align*}
     \EE[b^t_f &- b^{t+1}_f \mid e^t=e]
    = \EE\Big[\min\Big(b_f^t, \sum_{S \in \cS} A_{f S} \cdot \ind{S \text{ is sampled at time } t} \Big) ~\Big|~ e^t = e\Big]\\
    &\ge \nf{1}{168} \cdot \min\Big(b_f^t, \sum_{S \in \cS} A_{f S} \cdot \Pr[S \text{ is sampled at time } t \mid e^t = e]  \Big)
    = \nf{1}{168} \cdot \min\Big(b_f^t, \sum_{S \in \cS} A_{f
      S} \cdot \frac{p_S^t \; \kappa^t_e}{\optest}\Big) \\
    &\geq \nf{1}{168} \cdot  \frac{\kappa^t_e}{\optest} \cdot
      \min\Big(b_f^t, \ip{A_f,p^t}\Big)
      = \nf{1}{168} \cdot \frac{\kappa^t_e}{\density_f} \cdot \gain_{t,f}(p^t).
  \end{align*}
  To prove property~\ref{item:ass1} we substitute the above
  expression back into~(\ref{eq:6}) and recall $\sum_{e \in \cU^t} \kappa^t_e = \Phi^t$ to get
  \begin{align*}
    \realgain_t(p^t) &\ge \nf{1}{168} \cdot 
      \underset{f \sim \Suffix^t}{\EE}
                       [\gain_{t,f}(p^t) ]  = \nf{1}{168}\cdot \gain_t(p^t).
  \end{align*}

  (\emph{Bounded Gradients}): To prove property~\ref{item:ass3}, we
  observe that $\nf{\partial}{\partial p_S} \; \gain_{t,f}(p) = 0$ if
  $b^t_f$ is the minimum in the $\gain_{t,f}(p)$ expression. Otherwise,
  using that $\density_f \leq c_S/A_{fS}$, we get
  \begin{align*}
    \frac{\partial}{\partial p_S} \gain_{t,f}(p) &= \frac{\density_f}{\optest} \cdot A_{f S} \le \frac{c_S}{\optest}.
  \end{align*}

  (\emph{Sufficient Static Gains}): As in previous applications,
  define $p^\star$ to be the optimal (fractional) solution to the
  linear programming relaxation of (\ref{eq:ROCIP}). The argument is a familiar one:
  since $p^\star$ is feasible, we have $\ip{A_f, p^\star} \geq b_f$,
  and hence
  $\gain_{t,f}(p^\star) = b_f^t \density_f/\optest =
  \kappa_f^t/\optest$. Property~\ref{item:ass4} follows from taking
  expectations again: 
  \[ \gain_t(p^\star) = \frac{1}{\optest} \cdot \underset{f \sim \Suffix^t}{\EE}[\kappa_f^t] \geq
    \frac{\EE[c(\ALG_t)]}{2(\approxfactor+1)\, c(\OPT)}.\]
  The theorem now follows by \Cref{lem:gain-oco-application}, and the
  discussion in~\Cref{sec:half-cover-augm} showing that $\approxfactor
  = 4$, and that the half-coverage problem can be converted (online)
  into a solution to the CIP with only a constant factor loss. 
\end{proof}

\subsection{Non-Metric Facility Location}\label{sec:nmfl}

As the third and final application, we consider random-order
non-metric facility location. There are clients/elements $\cU$ and
facilities $\cS$, which are all points in a distance space (which is
not required to satisfy the triangle inequality). Opening a facility
at location $i \in \cS$ incurs a cost of $c_i$; each client
$e \in \cU$ can then be connected to some open facility $i$,
incurring a cost of $d_{ie}$. Formally, we want a solution to the
following IP:
\begin{align}
    \min \quad & \sum_{i \in \cS} c_iy_i + \sum_{i \in \cS}\sum_{e \in
                 \Suffix} d_{ie}x_{ie} \notag\\
    \text{s.t.} \quad & \sum_{i \in \cS} x_{ie} = 1 && \forall e \in
                                                       \Suffix
                                                       \tag{\textsc{NMFL}}
  \label{eq:ronmfl}\\
    \phantom{\text{s.t.}\quad} & x_{ie} \le y_i && \forall i \in \cS,
                                                   e \in \Suffix  \notag\\
    \phantom{\text{s.t.}\quad} & y_i, x_{ie} \in \{0, 1\} && \forall i
                                                             \in \cS,
                                                             e \in
                                                             \Suffix. \notag
\end{align}
Above, $y_i$ indicates that facility $i$ is opened, and $x_{ie}$ indicates that client $e$ is connected to facility $i$.
In the random-order
version of the problem, clients $e \in \Suffix$ arrive according to a
uniform permutation; we may open some facilities when a client
arrives, but then have to connect it to an open facility and pay the
connection cost.

\subsubsection{Augmentation Costs and the Potential}

Let $\cR^t$ represent the facilities opened by the algorithm before
time step $t$. The augmentation cost for a generic client $e$ at time
$t$ with respect to collection $\cR^t$ of facilities is:
\[
    \aug(e, \cR^t) = \min_{i \in \cS}(\ind{i \notin \cR^t} c_i +
    d_{ie}) = \min\big(\min_{i \in \cR^t} d_{ie}, \min_{i \not\in
      \cR^t} (c_i + d_{ie})\big). 
\]
That is, we can either choose the closest open facility in $\cR^t$, or
open a new facility and connect to it. We define the augmentation cost estimate to be the cost itself, namely $\kappa^t_e = \aug(e, \cR^t)$, which means the approximation factor in
equation~\eqref{eq:kappa-gen} is $\approxfactor = 1$.

Next, conditioned on the algorithm’s history $\cH^t$ up to time $t$, the expected real gain at step $t$ is given by:
\begin{equation}
  \realgain_t = \frac{1}{\Phi^t \cdot |\Suffix^t|} \sum_{e,f \in \Suffix^t} \underset{\coins^t}{\EE}\big[\kappa^t_f - \kappa^{t+1}_f \mid e^t=e\big].
  \label{eq:rg-nmfl}
\end{equation}

As we did in previous applications, we introduce the
linearized/decoupled version of the gain for the problem at hand so
that it fits into the general framework described in
Section~\ref{sec:framework}. To this end, let
$\Gamma^t_f := \{ i \in \cS \mid d_{if} \le
\nicefrac{\kappa^t_f}{2}\}$ denote the set of facilities which, once
opened, make the cost of satisfying $f$ half (or less) than its
current cost $\kappa^t_f$, and let $\chi^t_f \in \{0,1\}^m$ be the
indicator vector of the set $\Gamma^t_f$. We define the linearized/decoupled gain as:
\begin{align*}
  \gain_t(p) := \mathbb{E}_{f \sim \Suffix^t}\left[\gain_{t,f}(p)\right],
  \quad \text{where} \quad
  \gain_{t,f}(p) := \frac{\kappa^t_f}{2\optest}  \cdot \min\big(1,
  \ip{\chi^t_f, p} \big).
\end{align*}
Using this definition of
$\kappa_e^t$ and $\gain_t(p)$, we instantiate the framework $\augoco$
on the polytope
$\{ p \in [0,1]^n \mid \ip{c, p} \leq \optest \}$.

\subsubsection{Competitiveness}

\begin{theorem}\label{thm:ronmfl-oco}
    There exists an $O(\log (mn))$-competitive algorithm for random-order non-metric facility location.
\end{theorem}

\begin{proof}
  Let us fix a time $t$, and condition on the history $\cH^t$ until time $t$. We check all properties in \Cref{lem:gain-oco-application}, except property \ref{item:ass2} (unbiasedness), since this is satisfied by definition of the $\gain_t$ function.
  
  (\emph{Lower Bound on Real Gain}) Fix a client $f$ and observe that if in iteration $t$ the algorithm samples some facility in $\Gamma^t_f$, then $\kappa^{t+1}_f \le \nicefrac{\kappa^t_f}{2}$ by definition. Therefore,
  \begin{align*}
    \underset{\coins^t}{\EE}[\,\kappa^t_f &- \kappa^{t+1}_f \mid e^t = e] \ge \nicefrac{\kappa^t_f}{2} \cdot \Pr(\textrm{some $i \in \Gamma^t_f$ is sampled at time $t$} \mid e^t = e) \\
                              &= \nicefrac{\kappa^t_f}{2} \cdot \bigg(1 - \prod_{i \in \Gamma^t_f} (1- (\nf{\kappa^t_e}{\optest}) \cdot p^t_i)\bigg) 
                                \ge \frac{1- \nf1e}{2} \cdot \kappa^t_f \cdot \min\bigg(1,
                                (\nf{\kappa^t_e}{\optest}) \cdot \sum_{i \in \Gamma^t_f} p^t_i
                                \bigg) \\
                              &\ge \frac{1- \nf1e}{2} \cdot \kappa^t_f \cdot
                                \frac{\kappa^t_e}{\optest} \cdot \min\big(1, \ip{
                                \chi^t_f, p^t}  \big) = (1- \nf1e) \cdot \kappa^t_e
                                \cdot \gain_{t,f}(p^t),
  \end{align*}
  where the first equality uses that, conditioned on $e^t = e$, the
  sampling probability of facility $i$ is
  $(\nf{\kappa^t_e}{\optest}) \cdot p^t_i$. Substituting back, we
  get that 
  \begin{align*}
    \realgain_t \ge (1- \nf1e) 
    \cdot \gain_{t}(p^t),
  \end{align*}
  showing property \ref{item:ass1}.

  (\emph{Bounded Gradients}): We have that the subgradient
  $\nf{\partial}{\partial p_i} \, \gain_{t,f}(p) = 0$ if
  $\sum_{i \in \Gamma^t_f} p^t_i \ge 1$ or if $i \notin
  \Gamma^t_f$. Otherwise, if facility $i \in \Gamma^t_f$, we have
  $d_{if} \le \nf{\kappa^t_f}{2}$.  Moreover, since $\kappa^t_f$ is
  the minimal augmentation cost for client $f$, we get that
  $\kappa^t_f \le c_i + d_{if}$. Combining the two inequalities, we
  infer $\kappa^t_f \le 2c_i$ and therefore
  \begin{align*}
    \frac{\partial}{\partial p_i} \gain_{t,f}(p) &\le \frac{\kappa^t_f}{2\optest} \leq \frac{c_i}{\optest},
  \end{align*}
  which shows property \ref{item:ass3}.

  (\emph{Sufficient Static Gains}): To show this property, set
  $p^\star = y^\star$, i.e., the optimal fractional solution to the
  linear programming relaxation of~(\ref{eq:ronmfl}). 
  Let $\textsc{Easy}^t$ be the set of remaining clients
  $f \in \Suffix^t$ such that $\ip{\chi^t_f, p^\star} \geq \nf12$,
  i.e., the optimal fractional solution opens at least half a facility
  in $\Gamma^t_f$. This means
  \begin{align*}
    \gain_t(p^\star)
    &= \E{f \sim
      \Suffix^t}{\gain_{t,f}(p^\star)} \geq \frac{1}{
    |\Suffix^t|}\sum_{f \in \textsc{Easy}^t} \frac{\kappa^t_f}{2\optest} \cdot \min(1,
                                         \ip{\chi^t_f, p^t}) \\
      &\ge \frac{1}{2\optest \cdot
    |\Suffix^t|}\sum_{f \in \textsc{Easy}^t} \kappa^t_f\cdot\nicefrac{1}{2} =
    \frac{1}{4\optest \cdot |\Suffix^t|}\Bigg(\sum_{f \in \Suffix^t}
    \kappa^t_f -\sum_{f \in \textsc{Hard}^t} \kappa^t_f \Bigg),
  \end{align*}
  where $\textsc{Hard}^t := \Suffix^t \setminus \textsc{Easy}^t$. If
  $f \in \textsc{Hard}$, then in the optimal fractional solution at
  least half of $f$'s connection has to come from facilities outside
  $\Gamma^t_f$, i.e., at distance at least $\nicefrac{\kappa^t_f}{2}$.
  This means the optimal fractional connection cost for $f$ is at
  least $\nicefrac{\kappa^t_f}{4}$.

  Observe that $\EE_{f \sim \Suffix^t}[\kappa^t_f] \leq 2\EE[c(\ALG_t)]$, and so we get
  \begin{align*}
    \gain_t(p^\star) \geq \frac{\EE[c(\ALG_t)]}{8\optest} - 
    \frac{1}{4\optest \cdot |\Suffix^t|} \sum_{f \in \textsc{Hard}^t} \kappa^t_f. 
  \end{align*}
  By the discussion above, any client $f \in \textsc{Hard}^t$ has
  that its augmentation cost $\kappa^t_f$ is at most four times the fractional connection cost
  for $f$ in $\OPT$, and hence
  $\sum_{f \in \textsc{Hard}^t} \kappa^t_f \leq 4c(\OPT)$. Moreover, the
  cardinality of $\Suffix^t$ is $n-t+1$, and hence the above
  expression becomes
  \begin{align}
    \gain_t(p^\star) \geq \frac{\EE[c(\ALG_t)]}{8\optest} - 
    \frac{c(\OPT)}{\optest \cdot (n-t+1)}. 
  \end{align}
  Finally, the fact that
  $c(\OPT) \le \optest \le 2c(\OPT)$  gives 
  \begin{align}
    \gain_t(p^\star) \geq \frac{\E{}{c(\ALG_t)}}{16\,c(\OPT)} - \frac{1}{n - t + 1}.
  \end{align}
  Adding over all times $t < \tau \leq n$ and taking expectations,
  we get
  \[
    \E{}{\sum_{t < \tau} c(\ALG_t)} \leq 16\,c(\OPT) \cdot \bigg(
    \E{}{\sum_{t < \tau} \gain_t(p^\star)} + \log n \bigg).
  \]
  thereby establishing property \ref{item:ass4}. \Cref{thm:ronmfl-oco} is, thus, implied by \Cref{lem:gain-oco-application}.  
\end{proof}

\section{Closing Remarks}
\label{sec:closing-remarks}

In this work, we present a unified and modular framework connecting random-order online algorithms to online learning. By isolating the learning component, we show that any suitable OCO algorithm can be used as a black-box subroutine to achieve optimal competitiveness. This approach recovers $O(\log mn)$-competitive algorithms for classic problems—Weighted Set Cover, Unweighted Set Multicover, Covering Integer Programs, and Non-metric Facility Location—through a single, clean template.

\paragraph{Future Directions.} The framework outlined above appears conceptually general and opens several promising directions for future work. A primary challenge lies in extending it to broader classes of problems with box constraints---for example, unweighted \emph{multiset} multicover (i.e., unweighted set multicover where the matrix $A$ is not Boolean). For such settings, it remains unclear how to establish the requisite bounds, in particular, designing a function $\gain_{t,f}$ that both has bounded gradients (property~\ref{item:ass3}) and satisfies the static gains condition (property~\ref{item:ass4}). This is a key open question and a compelling direction for further understanding the workings and limitations of the technique.

\clearpage

\bigskip
\appendix
{\LARGE  \bf Appendix}

\section{Stochastic Online Mirror Descent}
\label{sec:stoc-OMD}

In this section, we give prove~\Cref{thm:regretConcaveScaled}, the
regret bound for online concave gain maximization, when at each
timestep, the feedback we get is a (bounded) unbiased estimator of the
gradient (restated for convenience).

\thmregretConcaveScaled*

While the ideas are relatively standard, we need
multiplicative-additive bounds for arbitrary subsets of scaled
simplices. Lacking a convenient reference, we give a proof for completeness.

\subsection{OMD Regret Bounds using Local Norms}
\label{sec:omd-regret-bounds}

The starting point is the following general result on regret minimization using Online Mirror Descent (OMD) from \cite[Lemma~6.33]{orabona-v7}. Consider convex loss
functions $\ell_t: \cX \rightarrow \R$ over a non-empty closed convex set
$\cV \subseteq \cX$ (assume $\cX$ is also a non-empty closed convex set). Consider a twice differentiable function $\psi : \cX \rightarrow \R$ with Hessian being positive definite in the interior of its domain, and let $B_\psi$ be the Bregman divergence w.r.t.. As in \cite[Assumption~6.5]{orabona-v7}, assume
$\lim_{\lambda \rightarrow 0} \ip{\nabla \psi(x + \lambda (y-x)), y-x}
= - \infty$ for all $x \in \text{bdry}(\cX)$ and
$y \in \text{int}(\cX)$. Given a square matrix $A$, define the $A$-norm to be
$\|x\|_A := \sqrt{x^\intercal A x}$.

\begin{lemma}[\cite{orabona-v7} OMD Regret via Local Norms]
  \label{lem:local-norm}
  Given the above setup, define
  \begin{gather}
    x_{t+1} \in \argmin_{x \in \cV} \ip{\partial \ell_t(x_t),
      x} + \frac{1}{\eta_t} B_{\psi}(x,x_t)
    \quad\text{and} \label{eq:OMD-iter1}  \\
    \tilde{x}_{t+1} \in \argmin_{x \in \cX} \ip{\partial \ell_t(x_t),
      x} + \frac{1}{\eta_t} B_{\psi}(x,x_t); \label{eq:OMD-iter2}
  \end{gather}
  suppose $x_{t+1}, \tilde{x}_{t+1}$ exist. 
  Then
  for every $u \in \cV$, there exists $\tilde{z}_t$ in the segment
  between $x_t$ and $\tilde{x}_{t+1}$ such that
  \begin{align}
    \ell_t(x_t) - \ell_t(u) \le \frac{1}{\eta_t} \big( B_{\psi}(u; x_t) -
    B_{\psi}(u; x_{t+1}) \big) + \frac{\eta_t}{2} \|\partial
    \ell_t(x_t)\|^2_{(\nabla^2 \psi(\tilde{z}_t))^{-1}}. \label{eq:ora-bound}
  \end{align}
\end{lemma}

Applying this with $\cX = \R^d_+$ and
$\cV \sse \fullsimplex = \{x \in [0,1]^d \mid \sum_i x_i \leq 1\}$ and
with $\psi$ being the unnormalized entropy function
$\psi(x) = \sum_i (x_i \ln x_i - x_i)$, we get the Bregman divergence
being the (unnormalized) Kullback-Liebler divergence
$\text{uKL}(p; q) = \sum_i p_i \log \nicefrac{p_i}{q_i} - p_i +
q_i$. This setting satisfies the conditions above
\Cref{lem:local-norm}. We can now
derive the following corollary:
\begin{corollary}[Subsets of the Full Simplex]
  \label{cor:regret}
  Suppose $\eta_t = \eta \in (0,1]$. Consider convex functions $\ell_t$ whose
  subgradients satisfy the $\ell_\infty$-boundedness condition
  $\|\partial \ell_t(x)\|_{\infty} \le 1$. Let $\cV \sse
  \fullsimplex$. Then the OMD updates given by~(\ref{eq:OMD-iter1})
  satisfy that for any $u \in \cV$,
  \begin{align*}
    \ell_t(x_t) - \ell_t(u) \le \frac{1}{\eta} \big(
    \textup{uKL}(u; x_t) - \textup{uKL}(u; x_{t+1}) \big)
    + O(\eta) \cdot \sum_i (\partial_i
    \ell_t(x_t))^2 \cdot x_{t,i}. 
  \end{align*}
\end{corollary}

\begin{proof}
  The Hessian is
  $\nabla^2 \psi(\tilde{z}_t) =
  \text{diag}(1/\tilde{z}_{t,i})$. Hence, for any $v$,
  $\|v\|_{(\nabla^2 \psi(\tilde{z}))^{-1}}^2 = \sum_i v_i^2 \cdot
  \tilde{z}_i$. Moreover, solving the KKT optimality conditions shows
  that the (unconstrained) minimizer $\tilde{x}_{t+1} = x_t \cdot
  e^{-\eta \partial \ell_t(x_t)} \leq e\cdot x_t$, since the gradients
  are bounded in $[-1,1]$ by assumption. Since $\tilde{z}_t$ lies on
  the line segment between $x_t$ and $\tilde{x}_{t+1}$, its coordinates are
  sandwiched between those of the two vectors, and hence 
  \[ \frac{\eta_t}{2} \|\partial \ell_t(x_t)\|^2_{(\nabla^2
      \psi(\tilde{z}_t))^{-1}} \leq \frac{e \eta}{2} \cdot \sum_i
    (\partial_i \ell_t(x_t))^2 \cdot x_{t,i}.
  \]
  The result then follows from \Cref{lem:local-norm}.
\end{proof}

When we are given a linear \emph{gain} function $\ip{a_t, x}$
for $a_t \in [0,1]^d$, we can define $\ell_t(x) = -\ip{a_t, x}$, and use
the fact that $\partial \ell_t(x) = - a_t$
and $a_{t,i}^2 \leq a_{t,i}$ to derive that for any $u \in \cV$,
\begin{align*}
  (1+O(\eta)) \ip{a_t, x_t} - \ip{a_t, u} \geq - \frac{1}{\eta} \big(
  \textup{uKL}(u; x_t) - \textup{uKL}(u; x_{t+1}) \big).
\end{align*}
Summing over all times $t$, and simplifying, 
\begin{align*}
  \sum_t \ip{a_t, x_t} \geq (1-O(\eta)) \sum_t \ip{a_t, u} - O(1) \cdot
  \frac{\textup{uKL}(u; x_0)}{\eta}.
\end{align*}
Choosing $x_0 = \nf1d \cdot \mathbf{1}$, and $u \in \cV \sse \fullsimplex$, the last term is at
most $O(\log d)$; suitably changing the constants in the value of
$\eta$, we get the familiar regret bound for linear gain functions: for every $T$,
\begin{align}
  \forall u \in \cV, \qquad \sum_{t \le T} \ip{a_t, x_t} \geq (1-\eta) \sum_{t \le T} \ip{a_t, u} -
  \frac{O(\log d)}{\eta}, \label{eq:gainOCO}
\end{align}
but now for constrained optimization over $\cV$ instead of over the full simplex.

\subsection{Extending to Stochastic Gradients}

We now extend to the setting where we get unbiased estimates of the
(sub)gradient, instead of getting the subgradient itself. In our
application, at each time $t$,
\begin{enumerate}
\item The algorithm plays some action $x_t \in \cV$, which depends on
  $\cH^t$, the history of everything that has happened at timesteps
  before $t$.
\item Then the adversary chooses a \emph{concave} gain function
  $g_t: \R^d \to \R_+$; this may depend on the history $\cH^t$ and
  also the algorithm's action $x_t$.

\item The algorithm sees a random vector $G_t \in \R^d$ whose expectation conditioned on the history $\cH^t$ and the action $x_t$ equals $\partial g_t(x_t)$. We assume that $G_t \in [0,1]^d$ with probability 1.

\end{enumerate}

Let $\tau$ be any stopping time adapted to the history sequence $\cH^1,\cH^2,\ldots$. The first observation is that for any particular sample path, the
bound from \eqref{eq:gainOCO} gives us that for any $x^* \in \cV$,
\begin{gather}
  \sum_{t \le \tau} \ip{G_t, x_t} \geq  (1-\eta) \, \sum_{t \le \tau} \ip{G_t,
    x^*} - \frac{O(\log d)}{\eta}. \label{eq:OLO-det}
\end{gather}
Since $\tau$ is a stopping time adapted to $(\cH^t)_t$ and $x_t$ is completely determined by the history $\cH^t$, we have $\E{}{\ip{G_t, x_t} \cdot \ones(t \le \tau) \mid \cH^t} = \ip{\E{}{G_t \mid \cH^t}, x_t} \cdot \ones(t \le \tau) = \ip{\partial g_t(x_t), x_t} \cdot \ones(t \le \tau)$, and similarly if we replace $x_t$ for $x^*$. Then adding \eqref{eq:OLO-det} over all $t$ and taking expectations we get that 
\begin{gather}
  \EE\Big[\sum_{t \le \tau} \ip{\partial g_t(x_t), x_t}\Big] \geq  (1-\eta) \, \EE\Big[\sum_{t \le \tau} \ip{\partial g_t(x_t),
    x^*}\Big] - \frac{O(\log d)}{\eta}. \label{eq:OLO-stoc}
\end{gather}
Note that we have expectations on both sides, since the choice of
$h_t$ is allowed to depend on the algorithm's actions, which are
themselves random.

But now, for the concave function $g_t$, we have the property that
\begin{gather*}
  g_t(x_t) - g_t(x^*) \geq \ip{ \partial g_t(x_t), x_t - x^*}, \quad \text{and} \\
  g_t(x_t) - g_t(0) \geq \ip{ \partial g_t(x_t), x_t}.
\end{gather*}
Using that $g_t(0) \geq 0$, multiplying the first inequality by
$(1-\eta)$ and the second by $\eta$, and summing, we get
\begin{gather}
  g_t(x_t) - (1-\eta)\, g_t(x^*) \geq \ip{\partial g_t(x_t), x_t} - (1-\eta)\,
  \ip{\partial g_t(x_t), x^*}. \label{eq:OLO-to-OCO}
\end{gather}
Combining this with \eqref{eq:OLO-stoc} gives that for any $x^* \in \cV$,
\begin{gather}
  \EE\Big[\sum_{t \le \tau} g_t(x_t)\Big] \geq  (1-\eta) \, \EE\Big[\sum_{t \le \tau} g_t(x^*)\Big] - \frac{O(\log d)}{\eta}, \label{eq:OCO-stoc}
\end{gather}
which is what we wanted to prove.

\subsection{Scaled Simplex}
\label{sec:scaled-simplex}

We now need to extend this to the action space 
$\scalesmpx := \{y \in [0,1]^d \mid \ip{c,y} \leq M\}$, and concave functions $h_t$ whose stochastic gradient estimates $H_t$ satisfy $H_{t,i} \in [0, \frac{c_i}{M}]$ for each coordinate $i \in [d]$. 

For that, one can just linearly transform the space. More precisely,
we define $\cV$ as the truncated simplex $\{x \in [0, \frac{c_i}{M}] :
\sum_i x_i \le 1\}$, which is just the scaling of the playing set
$\scalesmpx$ given by $x_i := \frac{c_i y_i}{M}$, and define the
function $g_t(x) := h_t(\frac{M x}{c})$ (where $\frac{M x}{c}$ is
defined coordinate-wise, i.e., its $i$th coordinate is $\frac{M
  x_i}{c_i}$). Using the chain rule, we have $\partial_i g_t(x) =
\frac{M}{c_i} \cdot \partial_i h_t(\frac{Mx}{c}) = \frac{M}{c_i} \cdot
\partial_i h_t(y)$, the last equation using the mapping between $x$- and $y$-space. Thus, the value $G_t := \frac{M}{c} \cdot H_t$ is an unbiased estimator of $\partial g_t(x_t)$ (given that $H_t$ is an unbiased estimator of $\partial h_t(y_t)$) and $G_t \in [0,1]^d$ (given that $H_{t,i} \in [0, \frac{c_i}{M}]$).

Let $y^\star := \frac{M x^*}{c}$; note that since $x^*$ ranges over all $\cV$, $y^\star$ ranges over all $\scalesmpx$.
The bound \eqref{eq:OCO-stoc} translates into the following guarantee:
\begin{align*}
    \forall y^\star \in \scalesmpx, \qquad \EE\Big[\sum_{t \le \tau} h_t(y_t)\big] \ge (1-\eta) \EE\Big[ \sum_{t \le \tau} h_t(y^\star) \Big]  - \frac{O(\log d)}{\eta}\,.
\end{align*}
This finally concludes the proof of \Cref{thm:regretConcaveScaled}.

\section{Probabilistic Inequalities}
\label{sec:prob-ineq}

The following claim is quite standard, but we give a proof here for
completeness. 
\begin{claim}
  \label{clm:1-1overe}
  For non-negative reals $a_1, a_2, \ldots$,
  we have $1 - \prod_{k} (1-a_k) \geq (1-\nf1e) \cdot \min(1, \sum_{k}
    a_k)$.
\end{claim}

\begin{proof}
  Without loss of generality, assume that $\sum_k a_k \leq 1$, else we
  can reduce some of $a_k$ values until this is satisfied; this only
  reduces the LHS without reducing the RHS. Now, 
  \[ 1 - \prod_k (1-a_k) \ge 1 - \exp\Big(-\sum_k a_k\Big) \ge (1-\nf{1}{e})
    \sum_k a_k, \] where the first inequality uses $1+y \leq e^{y}$
  for all reals $y$, and the second follows by minimizing
  $(1-e^{-y})/y$, for $y \in (0,1]$. 
\end{proof}

\begin{fact}[Fact 4.4 in \citep{GuptaKL21}]\label{fct:crs}   
    Let $\pi_j \in [0,1]$ be probabilities and $\lambda_j \in [0,1]$ be corresponding weights. Define
    \(
        \Lambda = \sum_j \lambda_j \cdot \textup{Ber}(\pi_j),
    \)
    as the sum of independent Bernoulli random variables scaled by the weights $\lambda_j$. Then, for any constant $C \ge \nf{1}{(e - 1)}$, it holds that
    \(
        \mathbb{E}\left[\min(\Lambda, C)\right] \ge \frac{1}{168} \cdot \min\left(\mathbb{E}[\Lambda], C\right).
    \)
\end{fact}

\section*{Acknowledgements}
Anupam Gupta is supported in part by NSF awards CCF-2224718 and CCF-2422926. Marco Molinaro is supported in part by the Coordenação de Aperfeiçoamento de Pessoal de Nível Superior - Brasil (CAPES) - Finance Code 001, and by Bolsa de Produtividade em Pesquisa $\#3$02121/2025-0 from CNPq. Matteo Russo is partially supported by the FAIR (Future Artificial Intelligence Research) project PE0000013, funded by the NextGenerationEU program within the PNRR-PE-AI scheme (M4C2, investment 1.3, line on Artificial Intelligence), and by the MUR PRIN grant 2022EKNE5K (Learning in Markets and Society). This work was done when M.R.\ was visiting New York University.

{\small
  \bibliographystyle{alpha}
  \bibliography{references,dblp}
}

\end{document}